\documentclass{article}
\usepackage{ijcai19}
\usepackage[export]{adjustbox}

\usepackage{times}
\usepackage{soul}
\usepackage{url}
\usepackage[hidelinks]{hyperref}
\hypersetup{
    colorlinks=true,
    linkcolor=red,
    citecolor=purple,
    urlcolor=cyan,
    filecolor=cyan
}

\usepackage[utf8]{inputenc}
\usepackage[small]{caption}
\usepackage{subcaption}
\usepackage[ruled,linesnumbered,noresetcount,vlined]{algorithm2e}

\SetKwFor{Repeat}{repeat}{:}{endw}
\usepackage{amsmath,amsfonts,mathrsfs,amssymb,bm,mathtools,mathabx}
\usepackage{graphicx, epsfig, epstopdf}
\usepackage{wrapfig, multirow, multicol,mdframed}
\usepackage{rotating, booktabs}
\usepackage{stmaryrd,xspace,arydshln}
\usepackage{picinpar}
\usepackage{todonotes}
\usepackage{color}
\usepackage{enumitem}
\usepackage{cleveref}

\usepackage{float}
\floatstyle{ruled}
\newfloat{model}{thp}{lop}
\floatname{model}{Model}

\newcommand{\citep}[1]{\citeauthor{#1}~[\citeyear{#1}]}

\def\thmspace{0.0em}
\newtheorem{theorem}{\hspace{\thmspace}{\bf Theorem}\!}
\newtheorem{corollary}{\hspace{\thmspace}{\bf Corollary}\!}
\newtheorem{definition}{\hspace{\thmspace}{\bf Definition}\!}

\newenvironment{proof}{{\textit{Proof}.}}{\hfill$\Box$}

\def \iff{\;\Leftrightarrow\;}

\newcommand{\benumerate}{\hspace{-0.5in} \begin{enumerate}\topsep=0pt \parsep=0pt \itemsep=-3pt } 
\newcommand{\eenumerate}{\end{enumerate}}

\newcommand{\bitemize}{\begin{list}{$\bullet$}{\topsep=0pt \parsep=0pt \itemsep=1pt \leftmargin=10pt}}
\newcommand{\eitemize}{\end{list}}

\newcommand{\argmin}{\operatornamewithlimits{argmin}}

\newcommand{\D}{{\mathscr{D}}}	
\newcommand{\R}{{\mathscr{R}}}	

\newcommand{\cA}{\mathcal{A}}

\newcommand{\cE}{\mathcal{E}}

\newcommand{\cM}{\mathcal{M}}
\newcommand{\cN}{\mathcal{N}}

\newcommand{\cO}{{\mathcal{O}}}	%

\newcommand{\cin}{\bm{G}}

\newcommand{\bx}{\bm{x}}

\newcommand{\bv}{\bm{v}}

\newcommand{\RR}{\mathbb{R}}
\newcommand{\NN}{\mathbb{N}}

\newcommand{\CC}{{\mathbb{C}}}

\newcommand{\Lap}{{\text{Lap}}}

\graphicspath{{images/}}
\SetAlgoSkip{}
\SetAlgoInsideSkip{}
\title{Privacy-Preserving Obfuscation of Critical Infrastructure Networks\\(Extended Version)}

\author{
Ferdinando Fioretto \and
Terrence W.K.~Mak \And
Pascal Van Hentenryck
\affiliations
Georgia Institute of Technology\\
\emails
\{fioretto, wmak\}@gatech.edu,
pvh@isye.gatech.edu
}

\begin{document}
\maketitle\sloppy\allowdisplaybreaks

\begin{abstract} 
The paper studies how to release data about a critical infrastructure
network (e.g., a power network or a transportation network) without
disclosing sensitive information that can be exploited by malevolent
agents, while preserving the realism of the network. It proposes a
novel obfuscation mechanism that combines several privacy-preserving
building blocks with a bi-level optimization model to significantly
improve accuracy. The obfuscation is evaluated for both realism and
privacy properties on real energy and transportation
networks. Experimental results show the obfuscation mechanism
substantially reduces the potential damage of an attack exploiting the
released data to harm the real network.
\end{abstract}

\section{Introduction}

\emph{Critical Infrastructure Networks} (CINs), such as electrical power grids and public transportation networks, rely on the tight interaction of cyber and physical components. They play crucial roles in ensuring social and economic stability, and their operations require advanced machine-learning and optimization algorithms. For instance, power network operations perform a power flow computation every few minutes.

Research on CINs is highly dependent on the availability of realistic
test cases. However, the release of such datasets is challenging due
to privacy and national security concerns. For instance, the
electrical load of an industrial customer may indirectly reveal
sensitive information on its production levels and strategic
investments. Similarly, the ability to link the physical and cyber
locations of power generators can be used to launch coordinated
attacks on cyber-physical facilities, significantly damaging the
targeted network.

To mitigate these concerns, this paper develops an obfuscation
mechanism based on {\em Differential Privacy} (DP) \cite{dwork:06}, a
robust framework that bounds the privacy risks associated with
answering sensitive queries or releasing datasets. A DP algorithm
introduces carefully calibrated noise to the data to prevent the
disclosure of sensitive information. It is immune to \emph{linkage
attacks}--attacks in which one exploits auxiliary data to expose
sensitive information.

However, DP faces significant challenges when the resulting
privacy-preserving datasets are used as inputs to complex optimization
algorithms on CINs. For instance, direct applications of DP may impact
the realism of the original dataset or produce networks that do not
admit feasible solutions for the problem of interest
\cite{fioretto:CPAIOR-18}.

\begin{figure}[!t]
  \centering
  \fbox{\includegraphics[width=0.8\linewidth]{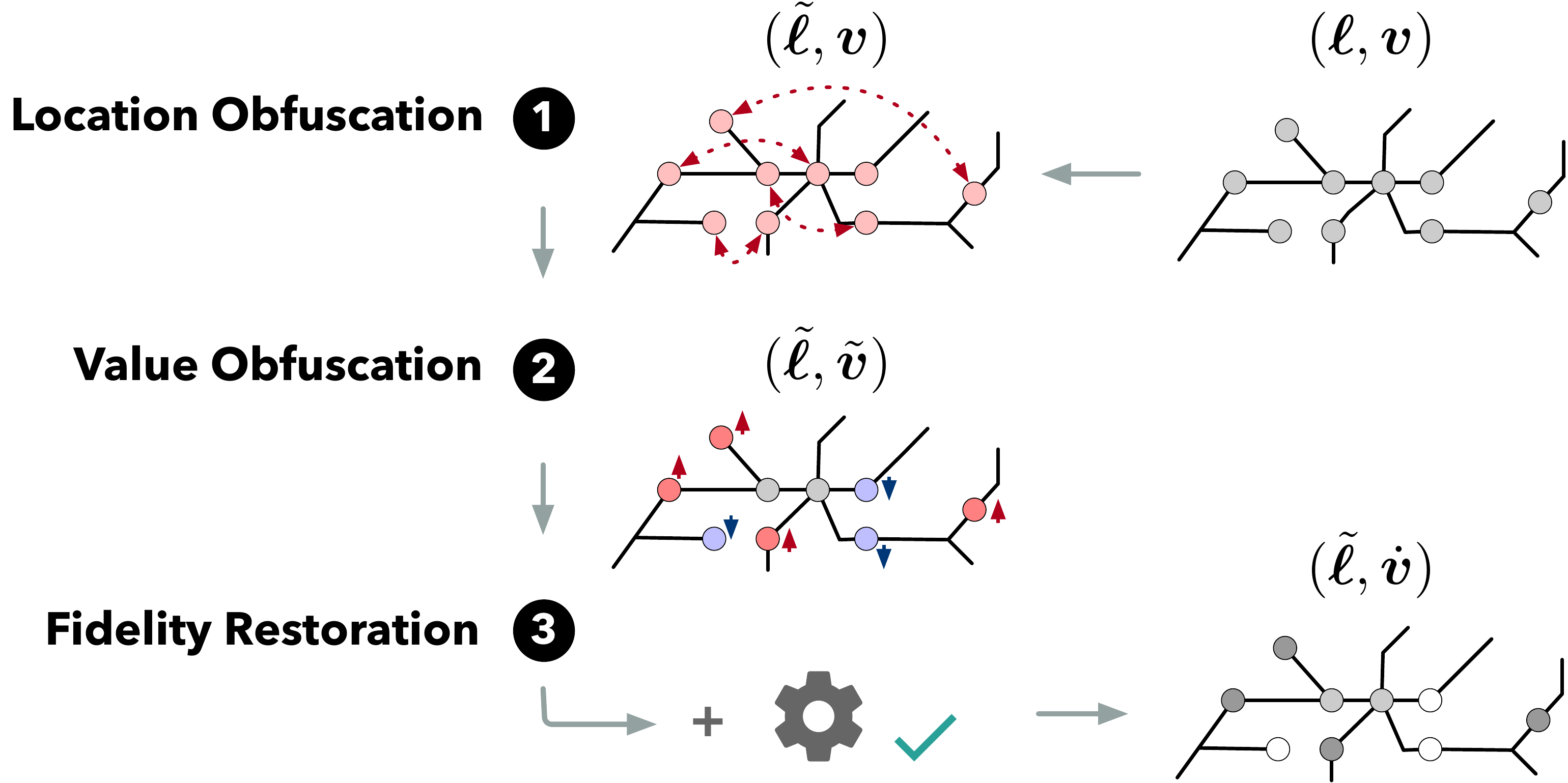}}
  \caption{The POCIN Framework
  }
  \label{fig:framework}
\end{figure}

The paper addresses this gap by introducing a \emph{Privacy-preserving
  Obfuscation mechanism for CINs} (POCINs), sketched in Figure
\ref{fig:framework}. POCIN takes as input a CIN and executes three
phases to (1) obfuscate the \emph{sensitive location} of network
elements, (2) hide \emph{sensitive values}, and (3) ensure the
satisfaction of important global properties of the CIN and the problem
of interest, through the use of optimization.

The paper makes the following contributions: (1) It proposes POCIN, a
novel privacy-preserving data release scheme that protects parameters
and locations of network elements; (2) POCIN uses a novel bi-level
optimization approach to preserve salient properties of the released
data and solves it with either exact or approximated methods; and (3)
It applies POCIN to two \emph{real} CINs from energy and
transportation networks. The paper shows that POCIN ensures strong
privacy guarantees and that, on the considered CINs, \emph{the damage
  inflicted on the real network, when the attacker exploits the
  obfuscated data, converges to that of a random uninformed attack} as
the privacy requirements increase.  


\section{Preliminaries: Differential Privacy}
\label{sec:differential_privacy}

\emph{Differential Privacy (DP)} is a framework used to protect the privacy of individuals in a dataset.
This notion has emerged as the de-facto standard for privacy
protection. A randomized mechanism $\cM \!:\! \D \!\to\! \R$ with
domain $\D$ and range $\R$ is $\epsilon$-differential private if, for
any output response $O \subseteq \R$ and any two \emph{neighboring}
inputs $D, D' \in \D$ differing in at most one individual (written
$D \sim D'$),
\begin{equation}
  \label{eq:dp_def} 
  Pr[\cM(D) \in O] \leq \exp(\epsilon)\, Pr[\cM(D') \in O],
\end{equation}

\noindent where the probability is calculated over the coin tosses of
$\cM$. The parameter $\epsilon \!>\! 0$ is the \emph{privacy loss} of the
algorithm, with values close to $0$ denoting strong privacy.  
DP satisfies several important
properties, including \emph{composability}, which ensures that a
combination of differentially private algorithms preserve differential
privacy, and \emph{immunity to post-processing}, which ensures that
privacy guarantees are preserved by arbitrary post-processing steps \cite{Dwork:13}.

In private data analysis settings, users interact with datasets by
issuing queries. A (numeric) \emph{query} is a function from a data
set $D \in \D$ to a result set $R \subseteq \RR^d$.  A query $Q$ can
be made differentially private by injecting random noise to its
output. The amount of noise depends on the \emph{sensitivity} of the
query, denoted by $\Delta_Q$ and defined as
\(
\Delta_Q = \max_{D \sim D'} \left\| Q(D) - Q(D')\right\|_1.
\)
In other words, the sensitivity of a query is the maximum
$l_1$-distance between the query outputs of any two neighboring
dataset $D$ and $D'$.

While the classical DP notion protects individuals from participating
into a dataset, many applications involve components whose {\em
presence} is public information. However, their values and (cyber)
locations are highly sensitive, as they may reveal, for instance, how
much power a generator is producing and where it is located, or the
most congested segments in transportation or logistics network.
Therefore the privacy goal of data curators is to protect observed
values and locations associated with these components.

The concept of \emph{indistinguishability} was introduced
by \citep{andres2013geo} to protect user locations in the Euclidean
plane and then generalized
in \cite{chatzikokolakis2013broadening,koufogiannis:15} to arbitrary
metric spaces.  Consider a dataset $D$ to which $n$ individuals
contribute their data $x_i$ and $\alpha > 0$. An adjacency relation
that captures the data variation of a single individual is defined as:
\begin{align*}
\label{eq:indist}
  D \sim_\alpha D' \iff \exists i: d(x_i, x_i') \leq \alpha \land\; \forall j \neq i: d(x_j, x_j') = 0
\end{align*}
where $d$ is a distance function on $\D$.  Such adjacency definition
is useful to hide individual participation up to some quantity
$\alpha$, or a location within a radius $\alpha$.
Given $\epsilon > 0, \alpha > 0$, a randomized mechanism $\cM : \D \to \R$ with domain $\D$ and range $\R$ is \emph{$\alpha$-indistinguishable $\epsilon$-DP} ($(\epsilon,\alpha$)-indistinguishable for short) if, for any event $O \subseteq \R$ and any pair $D \sim_\alpha D'$, ($D,D' \in \D$), Equation \eqref{eq:dp_def} holds.


\section{The CIN Obfuscation Problem}
\label{sec:obfuscation_problem}

\paragraph{Problem Setting}
\label{sec:problem_setting}

Consider a dataset $(\cal{N}, \cal{E})$ describing some critical
infrastructure network where $\cN \in \RR^{n \times p}$ and
$\cE \in \RR^{m \times q}$ describe the set of $n$ nodes and $m$
edges of the network, together with their attributes. The dataset
$(\cal{N}, \cal{E})$ is referred to as \emph{CIN description}.  An
element $n_i$ of $\cN$ is a $p$-dimensional vector describing the
attribute values associated with node $i$ and an element $e_{ij}$ of
$\cE$, is a $q$-dimensional vector describing the attribute values
associated with edge $(i,j)$ between nodes $i$ and $j$.  Furthermore,
consider an optimization problem $P$ that takes as input a CIN
description, denoted with $\cin$ for notational clarity:
\begin{equation}
  \begin{array}{rl}
    \cO^*(\cin) = \displaystyle\min_{\bx} 
    &\cO(\bx, \cin) \\
    \text{s.t.}
    &g_i(\bx, \cin) \leq 0 \;\; \forall i \in [k],
  \end{array}\!\!\!
  \label{c:bll}
  \end{equation}
where $\bx \!\in\! \RR^l$ is a vector of decision variables, $\cO$ is the objective function of $P$, and $g_i$ ($i \!\in\! [k]$) are the problem
constraints.

\begin{figure}[!t]
  \centering
  \includegraphics[width=0.65\linewidth]{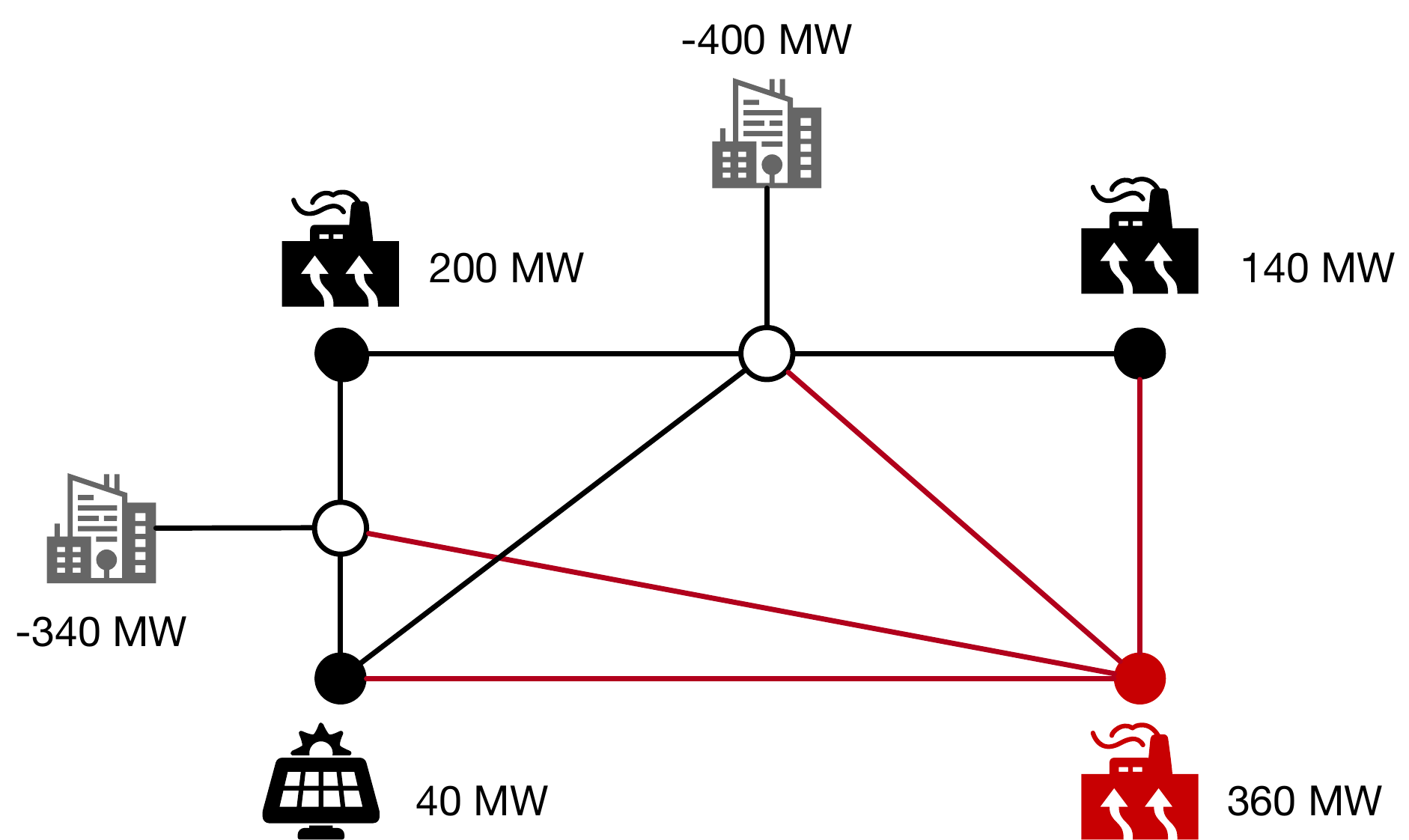}
  \caption{An example CIN with four nodes (power generators). The highlighted node is one being compromised by an attacker.}
  \label{fig:generator}
\end{figure}

To illustrate these concepts, consider the CIN in
Figure \ref{fig:generator}, representing a simplified \emph{power
network}. The nodes describe the network \emph{generators}
and \emph{loads} (represented as cities in the figure) with values
denoting the amount of power being injected into, or withdrawn from,
the power grid. The edges represent the transmission lines to carry
power from generators to loads. The optimization problem $P$ may be
the \emph{optimal power flow} (OPF) that amounts to finding the most
cost-effective generator dispatch to serve the load demands while
satisfying the physical constraints of the network. The data curator
desires to release a CIN description
$(\cal{\tilde{N}}, \cal{\tilde{E}})$ that \emph{obfuscates}
the \emph{location} and \emph{value} of some \emph{sensitive} network
elements, e.g., the locations and capacities of generators. 
The curator also wants to ensure that the OPFs on the released and
original CINs behave similarly (e.g., are feasible and have similar
costs) so that the released data is of high {\em fidelity}.

For notational simplicity, this paper focuses on
obfuscating \emph{sensitive nodes} of a CIN where each node
$n_i \in \cN$ is a pair $(\ell_i, v_i) \in \NN \times \RR$, describing
its location and salient parameter. The $n$-dimensional vectors of
locations and values are denoted by $\bm{\ell}$ and $\bm{v}$ and
$\cin$ is used to denote a CIN description.

\paragraph{Attack Model}
\label{sec:attack_model}

The paper considers an \emph{attack model} $\cA$ in which a malicious
agent can disrupt up to $b$ elements (called \emph{attack budget}) to
damage the network. Such action produces a new damaged network $\cin^a
= \cA(\cin)$ with $\cin^a \subseteq \cin$. For instance,
Figure \ref{fig:generator} highlights in red the generators affected
by an attack with a budget $b=1$, while $\cin^a$ is shown in black. The paper further assumes that the attacker has full knowledge of the
network topology and is capable of solving problem $P$ (e.g., to find
the generators with highest dispatch) and assessing the damages resulting
from an attack. The \emph{damage} of the attacker is measured using
the objective value difference $\cO^*(\cin^a) - \cO^*(\cin)$ for
problem $P$. For simplicity, the paper focuses on \emph{ranking
attacks}, i.e., attacks that disrupt the $b$ highest-ranked network
elements when their values are sorted according to an ordering
$\prec$ (e.g., the most dispatched generators), and evaluates
the released data against such attacks.

\paragraph{The CIN Obfuscation Problem}
\label{sec:CINOP}

This paper uses $\alpha$-indistinguishability to obfuscate the values
and locations of nodes. For a given $\alpha_\ell > 0$, the relation
$\sim_\ell \subseteq \NN^{2}$ captures indistinguishability between
node positions in the network, using a distance $d_\ell$ defined as a
function of the minimum number of nodes separating two nodes.
Similarly, for a given $\alpha_v > 0$ and distance function $d_v$,
relation $\sim_v \subseteq \RR^{2}$ captures the indistinguishability
between node values. These neighboring relations can be combined into
a new relation $\sim_{\ell v}$ that captures indistinguishability for
both values and locations. Two CIN descriptions $\cin \!=\!
(\bm{\ell}, \bv), \cin'\!=\!(\bm{\ell}', \bv') \in \NN \times \RR$ are
\emph{location and value indistinguishable} if and only if
\begin{equation}
    \cin \sim_{\ell v} \cin' \iff
    \ell \sim_\ell \ell' \lor \bv \sim_v \bv'.
\end{equation}
\noindent
A randomized mechanism $\cM$ is
$(\epsilon, \alpha_\ell, \alpha_v)$-indistinguishable if, for any pair
of $\sim_{\ell v}$-adjacent datasets, \Cref{eq:dp_def}
holds. 

The design of obfuscation mechanisms for CIN descriptions
should satisfy three desiderata, formalized through
the \emph{Privacy-preserving Obfuscation for CIN} (POCIN) problem:

\begin{definition}
  \label{def:pocin}
    Given a CIN description $\cin$, a problem $P$, and positive real values $\alpha_v, \alpha_\ell, \beta$ and $\epsilon$,
    the POCIN problem produces an obfuscated CIN description $\tilde{\cin} = (\tilde{\bm{\ell}}, \tilde{\bv})$ such that:
    \begin{enumerate}[leftmargin=*, parsep=0pt, itemsep=0pt, topsep=0pt]
      \item \emph{Privacy}: 
      $\tilde{\cin}$ satisfies $(\epsilon, \alpha_\ell, \alpha_v)$-indistinguishability. 

      \item \emph{Fidelity}: 
      $\tilde{\cin}$ admits a candidate solution $\bar{\bx}$ that \emph{(i)} satisfies the constraints $g_i(\bar{\bx}, \tilde{\cin})$
      of $P$ and  \emph{(ii)} is \emph{$\beta$-faithful} to the P's objective, i.e.,
      $\frac{| \cO(\bar{\bx},\tilde{\cin}) - \cO^*(\cin)|}  {\cO^*(\cin)} \leq \beta$.

      \item \emph{Robustness}:
       $\tilde{\cin}$ minimizes $|\cO^*(\cin) - \cO^*(\cA(\tilde{\cin}))|$, i.e., the damage inflicted by attack $\cA$.
    \end{enumerate}
  \end{definition}

\section{The POCIN Mechanism}
\label{sec:POCIN_mechanism}

The POCIN mechanism is divided in three phases, as illustrated in Figure \ref{fig:framework}:
\begin{enumerate}[leftmargin=*, parsep=0pt, itemsep=0pt, topsep=0pt]
  \item The \emph{location obfuscation} phase produces a new location indistinguishable CIN ${\cin}^\ell = (\tilde{\bm{\ell}}, \bv)$ from the original $\cin$. 

  \item The \emph{value obfuscation} phase takes ${\cin}^\ell$ as input and produces a new value indistinguishable CIN $\tilde{\cin} = (\tilde{\bm{\ell}}, \tilde{\bv})$. 

  \item The \emph{fidelity restoration} phase produces a new CIN description $\dot{\cin} = (\tilde{\bm{\ell}}, \dot{\bm{v}})$ that satisfies the problem constraints and is faithful to its objective.
\end{enumerate}

\subsubsection{Phase 1: Location Obfuscation} 

The first phase provides location indistinguishability. The idea is to shuffle the node locations using
an instance of the \emph{Exponential Mechanism} \cite{mcsherry:07}.
Such mechanism releases a privacy-preserving answer to a query by
sampling from its output discrete space $O$. The sampling probability
 is determined by a \emph{utility function} $u: (\D \times O) \to \RR $ that assigns a score to each output $o \in O$.

\begin{theorem}[Exponential Mechanism]
\label{th:m_exp}
  Let $u \!:\! (\D \times O) \!\to\! \RR$ with sensitivity 
  $
  \Delta_u \!=\! \max_{o \in O} \max_{D \sim D'} |u(D, o) \!-\! u(D', o)|.
  $ 
   The exponential mechanism that outputs $o$ with probability 
  $\Pr[o$ is selected $] \propto \exp \big( {\epsilon u(D, o)}/{2\Delta_u} \big)$
  satisfies $\epsilon-$DP.
\end{theorem}

POCIN uses a privacy-preserving \emph{shuffling function} that maps each node $i$ with location $\ell_i$ to a new location $\tilde{\ell}_i$ as:
\begin{equation}
\label{eq:exp_step2}
\tilde{\ell}_i \gets \ell_j \text{ with } \Pr \propto 
\exp\left( \frac{\epsilon\, d_\ell(\ell_i,\, \ell_j)}{2 \alpha_\ell} \right)
\end{equation}    
where $\ell_j$ is the location of a node $j$. This is an application of the Exponential mechanism with $\Delta_u$ being the sensitivity of distance function $d_\ell$. Since two
different locations may be mapped to the same location, POCIN solves
an \emph{assignment problem} so that each node $i$ is mapped to a
unique location. 

\begin{figure}[!t]
\centering
\includegraphics[width=0.65\linewidth]{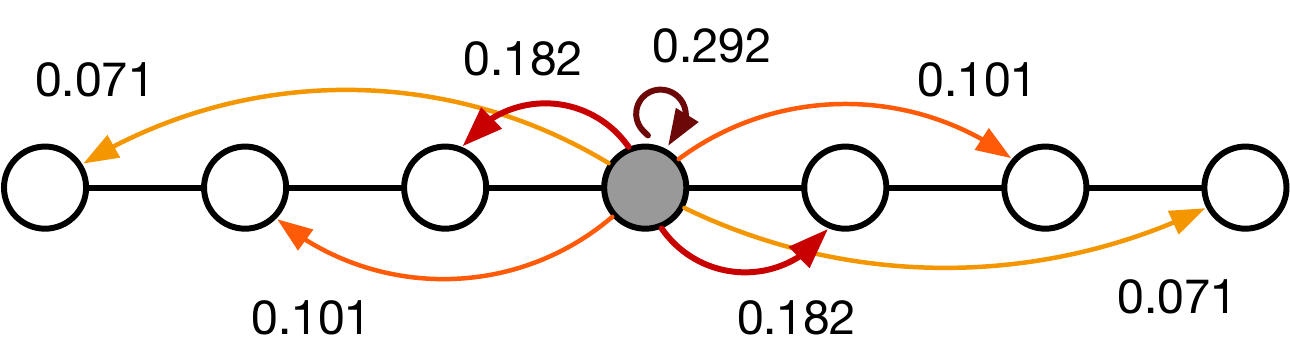}
\caption{Probabilities of location obfuscations ($\alpha_\ell\!=\!\epsilon\!=\!1.0$).\label{img:prob}}
\end{figure}

  \begin{corollary}
    \label{thm:phase2}
    Given $\alpha_\ell > 0$, Phase 1 of POCIN achieves $\alpha_\ell$-location-indistinguishability. 
  \end{corollary}
\noindent
Figure \ref{img:prob} illustrates the intuition underlying the
mechanism: It displays the probabilities of moving the central node to
a specific location. For a sufficiently large $\alpha_\ell$, an
attacker is provided with no useful location information. Of course,
there is a tradeoff between $\alpha_\ell$ and the fidelity of the
network, which is addressed in Phase 3.

\subsubsection{Phase 2: Value Obfuscation}

The second phase of POCIN takes as input ${\cin}^\ell \!=\!
(\bm{\tilde{\ell}}, \bm{v})$ and constructs a privacy-preserving answer
$\tilde{\bm{v}}$ for the node values $\bm{v}$ using the Laplace
mechanism.
  
\begin{theorem}[Laplace Mechanism]
\label{th:m_lap} 
Let $Q$ be a numeric query that maps datasets to $\RR^d$. The Laplace mechanism that outputs $Q(D) + z$, where $z \in \RR^d$ is drawn from the Laplace distribution
$\textrm{\Lap}(\Delta_Q / \epsilon)^d$, achieves
$\epsilon$-DP.
\end{theorem}

\noindent
In the above, $\Lap(\lambda)$ denotes the Laplace distribution with 0
mean and scale $\lambda$, and $\Lap(\lambda)^d$ denotes the
i.i.d.~Laplace distribution over $d$ dimensions with parameter
$\lambda$. The Laplace mechanism with parameter
$\lambda=\alpha/\epsilon$ satisfies
$\alpha$-indistinguishability \cite{chatzikokolakis2013broadening}. As a result, the privacy-preserving values $\tilde{\bv}$ of
the CIN description are obtained as follows:
\begin{equation} \label{eq:setp1}
\tilde{\bm{v}} = \bm{v} + \text{Lap}(\alpha_v / \epsilon)^n.
\end{equation}
The Laplace mechanism has been shown to be \emph{optimal}: it
minimizes the mean-squared error for identity queries with respect to
the L1-norm \cite{koufogiannis:15}.
\begin{corollary}
  \label{thm:phase1}
  Given $\alpha_v > 0$, Phase 2 of POCIN achieves $\alpha_v$-value-indistinguishability. 
\end{corollary}
Phase 2 generates a CIN description $\tilde{\cin} =
(\tilde{\bm{\ell}}, \tilde{\bv})$ obfuscating locations and values
of the original network $\cin$. As a result, $\cin^v$
satisfies \emph{condition (1)} of the CIN obfuscation problem.

\subsubsection{Phase 3: Fidelity Restoration}

The noise introduced by Phases 1 and 2 may significantly alter the
structure of the data and the solutions to Problem $P$.  To restore
the fidelity of the network, POCIN leverages the post-processing
immunity of DP and uses a bi-level optimization problem to
redistribute the noise introduced in the previous phases.  The primary decision variables of the problem are the vector
$\dot{\bm{v}} = (\dot{v}_1, \ldots, \dot{v}_n)$ that represents the
post-processed node values after the noise redistribution.

The problem, called $P_{\text{BL}}$, is shown in Figure
\ref{fig:biopt}. It searches for a vector $\dot{\bm{v}}$ for which 
problem $P$ has an optimal solution $\bx^*$ and whose objective
value is close to the original optimum $\cO^*$ (assumed to be public
information). Moreover, the vector $\dot{\bm{v}}$ must be as close as
possible to obfuscated vector $\tilde{\bv}$, which is ensured by
objective \eqref{bi:obj}. Constraint \eqref{bi:subproblem} ensures
that $\bx^*$ is an optimal solution to problem
$\cO^*(\dot{\cin}=(\tilde{\bm{\ell}},\dot{\bv}))$, 
and Constraint
\eqref{bi:constr1} ensures the fidelity of the objective.

\begin{figure}[t]
\begin{mdframed}
  \vspace{-12pt}
    \begin{align}
    \!\!\!\!\!
  P_{\text{BL}} \!=\!
    \min_{(\dot{\bv}, \bx)} & \| \dot{\bv} - \tilde{\bv} \|_2 \tag{b1} \label{bi:obj} \\
    \text{s.t.:}\; & 
      | \cO(\bx^*, \dot{\bv}) - \cO^*| \leq \beta
    \tag{b2} \label{bi:constr1} \\
  &\!\!\!\begin{array}{rl} 
    \bm{x}^* \!=\! \argmin_{\bx} 
      &\!\!\!\!\cO(\bx, \dot{\bv}) \\
    \text{s.t.}\;
      &\!\!\!\!g_i(\bx, \dot{\bv}) \leq 0 \; \forall i \in [k]\\
  \end{array}
  \tag{b3} \label{bi:subproblem}
\end{align}
\vspace*{-10pt}
\end{mdframed}
  \caption{The bi-level optimization of the Fidelity Restoration.}
  \label{fig:biopt}
\end{figure}  

\begin{theorem}
POCIN is $(\epsilon, \alpha_p, \alpha_v)$-indistinguishable. 
\end{theorem}
 

\begin{theorem}
\label{theorem:factor2}
The error induced by POCIN on the CIN node values is bounded by the inequality:
  $
  \| \dot{\bv} - \bv \|_2 \leq 2
  \| \tilde{\bv} - \bv \|_2.
  $
\end{theorem}

\noindent Note that POCIN does not explicitly minimize the damages
inflicted by an attack on the critical infrastructure network. Indeed,
simply evaluating the damage of an attack requires access to the real
network $\cin$, which would violate privacy. However, the shuffling
mechanism indirectly addresses this issue. Intuitively, for a
sufficiently large $\alpha_\ell$ in the location obfuscation, the
probability of choosing the most dispatched generator with a single
attack will be close to $\frac{1}{n}$.

\section{Solving the Fidelity Restoration Model}
\label{sec:solving}

$P_{\text{BL}}$ is a \emph{bi-level problem} and bi-level programming
is known to be strongly NP-hard \cite{hansen:92}: Even
evaluating a solution for optimality is NP-hard \cite{vicente:94}.
This paper explores two avenues to solve or approximate $P_{\text{BL}}$.

When the subproblem \eqref{bi:subproblem} is convex, it can be
replaced by its \emph{Karush-Kuhn-Tucker} (KKT) conditions, producing
a single-level model. Moreover, when the subproblem is linear, the
resulting subproblem is a MIP model, which is the case in our
transportation case study.
The subproblem can also be replaced by the {\em relaxation} of
$P_{\text{BL}}$ shown in Figure \ref{fig:nonconvex}. This relaxation
ensures that $\dot{\cin}$ has a feasible solution whose value is close
to the original objective. However, it does not constrain the optimal
solution which will be a lower bound to $\cO(\bx^*, \dot{\bv})$. This
relaxation is used in the power network case study and also satisfies
Theorem \ref{theorem:factor2}. When the subproblem is convex, it is
possible to strengthen Theorem \ref{theorem:factor2} using the
existence of a solution (vector $\bv$), the optimality of $\dot{\bv}$,
and the angle property of a projection on a convex set.

\begin{theorem}
\label{theorem:factor}
When $P_{\text{CL}}$ is used for fidelity restoration and is convex, 
the error induced by POCIN on the CIN node values is bounded by the inequality:
  $
  \| \dot{\bv} - \bv \|_2 \leq
  \| \tilde{\bv} - \bv \|_2.
  $
\end{theorem}

Other methods to solve bi-level problems include \emph{descent
methods} \cite{kolstad1990derivative}, \emph{penalty function methods}
\cite{aiyoshi1984solution}, and \emph{evolutionary approaches}
\cite{mathieu1994genetic}. The reader can consult \cite{sinha2018review}
for an extensive review on the topic.

\begin{figure}[t]
\begin{mdframed}
  \vspace{-12pt}
    \begin{align}
    \!\!\!\!\!
  P_{\text{CL}} \!=\!
      \min_{(\dot{\bv}, \bx)} & \| \dot{\bv} - \tilde{\bv} \|_2 \tag{n1} \label{ni:obj} \\
    \text{s.t.:}\; & 
      | \cO(\bx, \dot{\bv}) - \cO^*| \leq \beta
    \tag{n2} \label{ni:constr1} \\
      &\!\!\!\!g_i(\bx, \dot{\bv}) \leq 0 \; \forall i \in [k]
    \tag{n3} \label{bi:subproblem-r}
\end{align}
\vspace*{-16pt}
\end{mdframed}
  \caption{The Relaxation of $P_\text{BL}$.}
  \label{fig:nonconvex}
\end{figure}  

\section{Experimental Results}
\label{sec:experimental_results}

This section presents experimental results on two case studies: power
network and traffic obfuscation problems.

\subsection{Power Network Obfuscation Problem}

This scenario considers a transmission operator who would like to release a description of its network to stimulate research but is worried that malicious actors could use it to design a cyber-attack targeting the generators that would cause maximum damage. The operator seeks to obfuscate the locations and maximum capacities of its generators while preserving the realism of the network description. The experiments were performed on a variety of benchmarks from the NESTA library~\cite{Coffrin14Nesta}. 
The results analyze the dispatch values produced by POCIN and determine how well the obfuscated networks sustain an attack. For brevity, the results are highlighted on the IEEE 118 bus test case. 
The experiments use a privacy loss $\epsilon$ of $1.0$, and vary the \emph{indistinguishability   levels} $\alpha_l$ from 1\% to 10\% of the network diameter $d(\cin)$, 
and the \emph{faithfulness level} $\beta$ in $\{10^{-2}, 10^{-1}\}$,
while $\alpha_v$ is fixed to 0.1 p.u. $\approx$ 10 MW. The model was implemented using the Julia package PowerModels.jl with IPOPT~\cite{Coffrin:18} for solving the nonlinear AC power flow models. POCIN takes less than 1 minute to produce the obfuscated networks.

\begin{figure*}[!t]
\begin{subfigure}{.7\textwidth}
  \centering
         \includegraphics[width=.48\textwidth, valign=t]{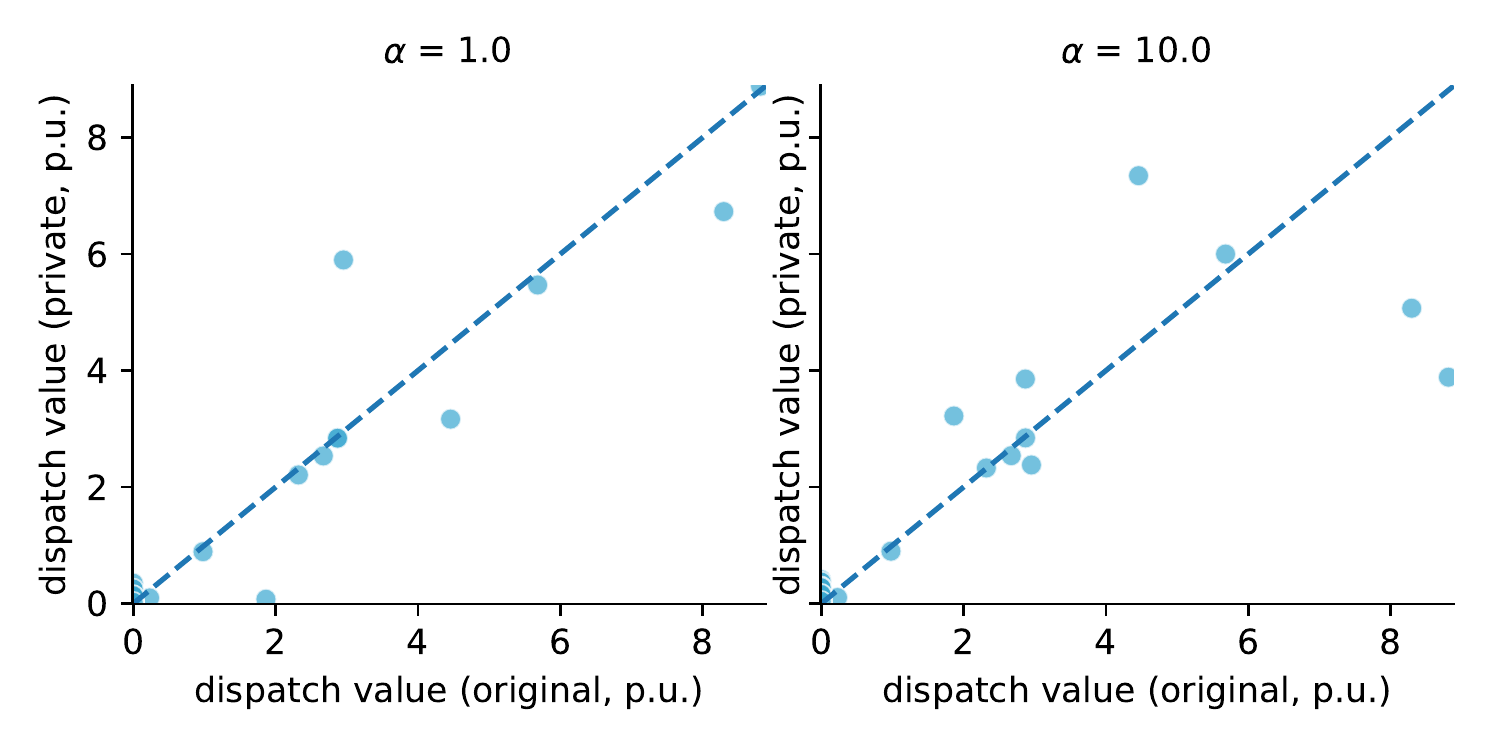}
         \includegraphics[width=.48\textwidth, valign=t]{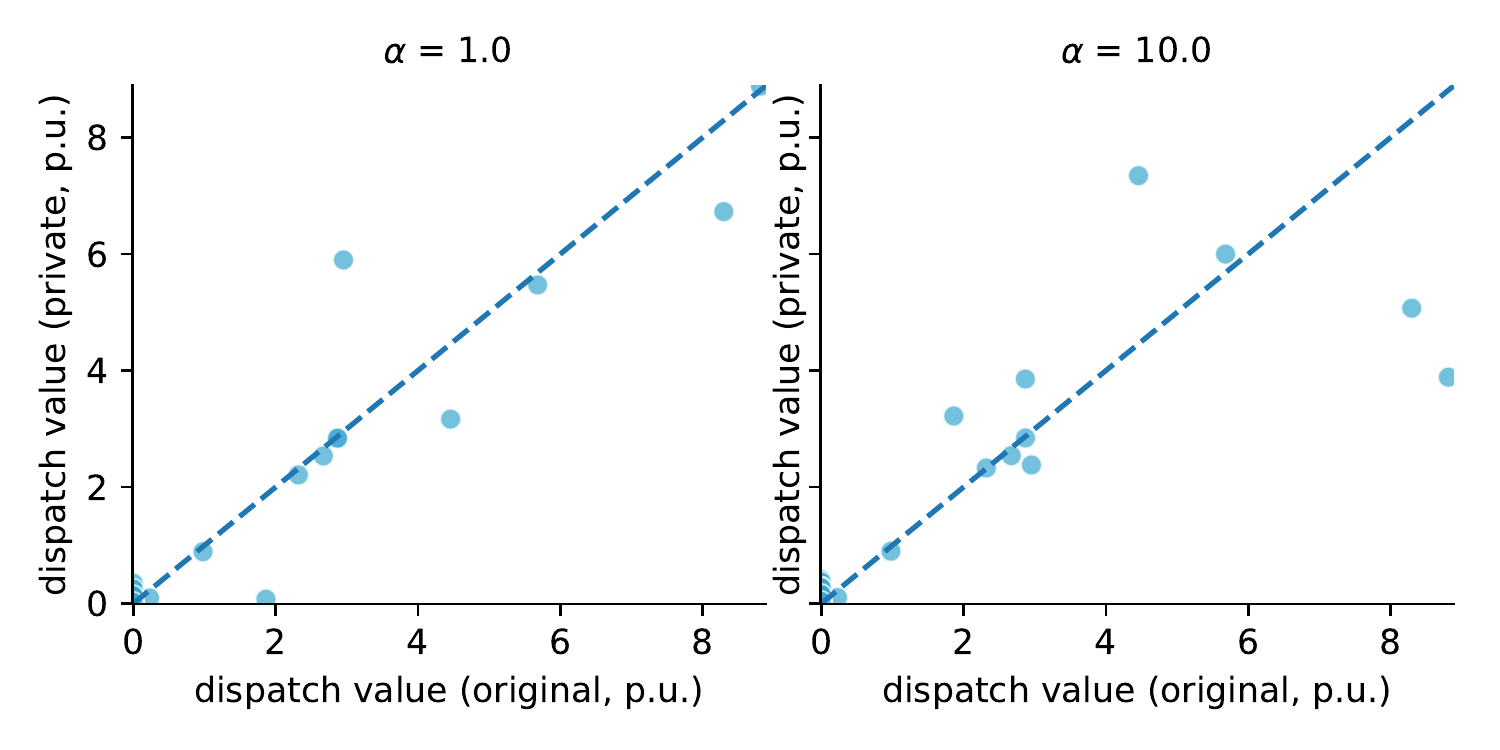}
  \caption{Active dispatch between original and obfuscated networks with $\epsilon=1.0$, $\beta=0.01$, \\$\alpha_l=\alpha \% \times d(\cin)$, $\alpha_v=0.1$. Left: Generator active generation. Right: Bus active injection. } 
  \label{fig:dispatch_118-gen}
\end{subfigure}
\begin{subfigure}{.30\textwidth}
  \centering
  \includegraphics[width=.70\textwidth,height=80pt]{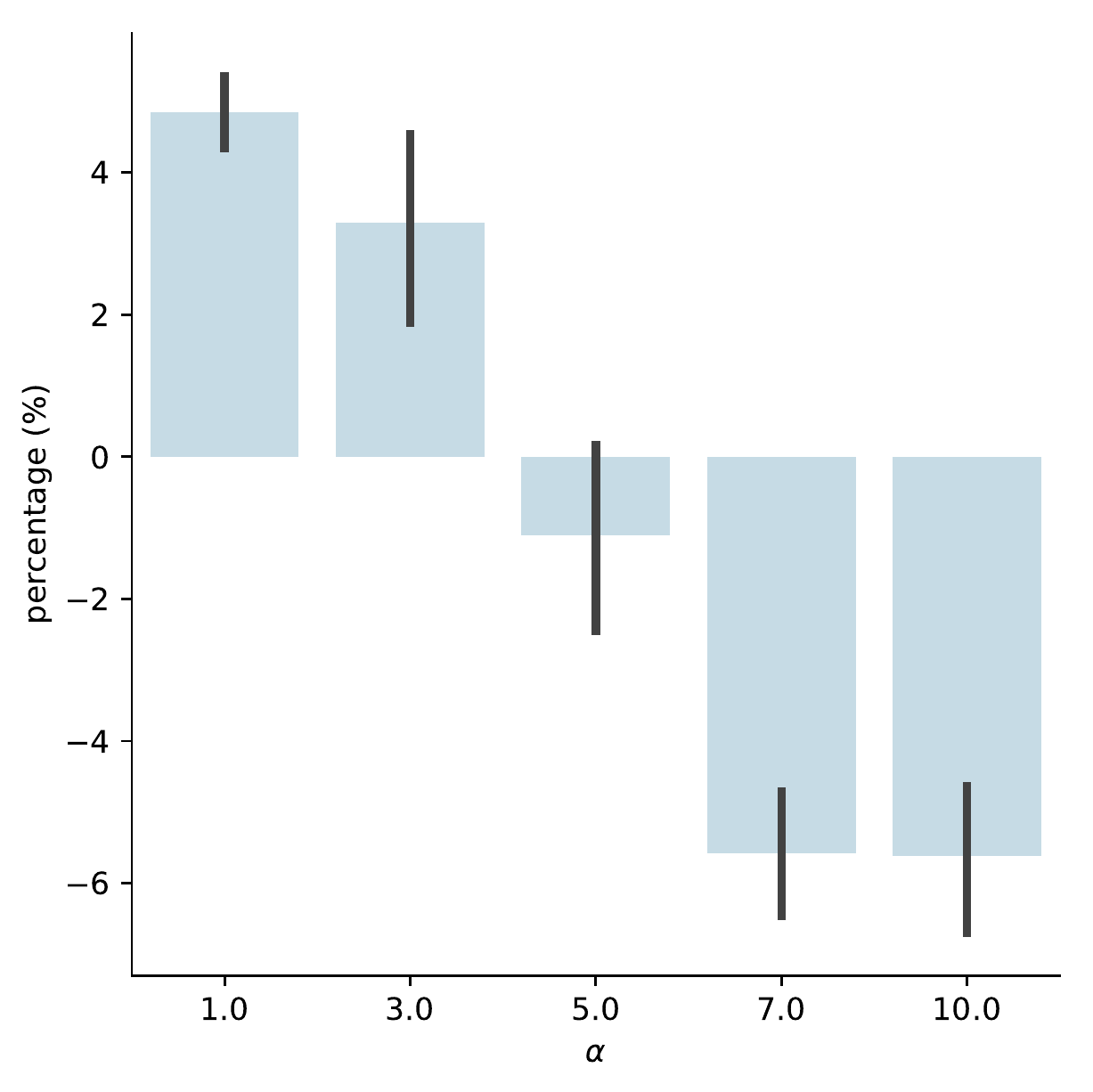}
  \caption{AC optimal dispatch costs differences (in \%) between original and obfuscated networks, with $\beta=0.1$}
  \label{fig:opf_118}
\end{subfigure}
\vspace{-2mm}
\caption{(IEEE 118 bus) Dispatch \& Cost Analysis}
\end{figure*}

\paragraph{Dispatch \& Cost Analysis}

The section studies the privacy and realism of the obfuscated networks
produced by POCIN.  \Cref{fig:dispatch_118-gen} shows the active
dispatch values of all generators in the obfuscated network and
compares them with their associated values in the original network.
The left plot compares the dispatch for each generator, while the
right plot compares the generation on each bus. Note that POCIN only
swaps generator locations and do not assign generators to arbitrary
buses. Not surprisingly, the dispatch differences are more pronounced
as the indistinguishability level increases. {\em More importantly,
  the dispatch values can no longer be used to distinguish generators
  and/or buses.}

\begin{table}[!t]
\centering
\small
\resizebox{0.6\linewidth}{!}
  {
  \begin{tabular}{lrrrrr}
  \toprule
  & \multicolumn{5}{c}{$\alpha$}\\
  \cmidrule(lr){2-6} 
  Network instance &   {1} & {3} & {5} & {7} & {10}\\
  \midrule
  nesta\_case14\_ieee  &      6   &   6   &   6   &   6   & 6\\
  nesta\_case30\_ieee      &    0   &   0   &   0   &   0   & 0 \\
  nesta\_case57\_ieee      &    12  &   10  &   6   &   12  & 12\\
  nesta\_case118\_ieee     &    0   &   0   &   0   &   0   & 0\\
  \bottomrule
\end{tabular}
}
\caption{Feasibility Statistics before Phase 3 (in percentage).}
\label{tbl:feasibility} 
\end{table}


Table~\ref{tbl:feasibility} reports the success rate of obtaining a
feasible AC power flow after Phase 2 on various benchmarks over 50
runs. Without Phase 3, the optimization problem $P$ is infeasible in most or all runs for all test cases.  {\em These results highlight the critical role of Phase 3 of POCIN.} \Cref{fig:opf_118} shows the difference, in percentage, between the dispatch costs of the obfuscated and original networks for increasing the indistinguishability levels $\alpha_l : \alpha_l=\alpha \% \times d(\cin)$.  It reports the mean and standard deviation (shown with black, solid, lines) over 50 runs. {\em The results indicate that obfuscated networks are largely
  within the faithfulness requirement and provide the necessary
  fidelity for Problem $P$ on the obfuscated network.}

%
%
%

\begin{figure}[t]
  \centering
  \includegraphics[width=.30\textwidth]{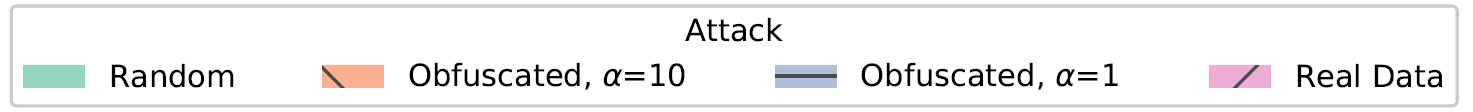}\\
   \includegraphics[width=.20\textwidth,height=85pt]{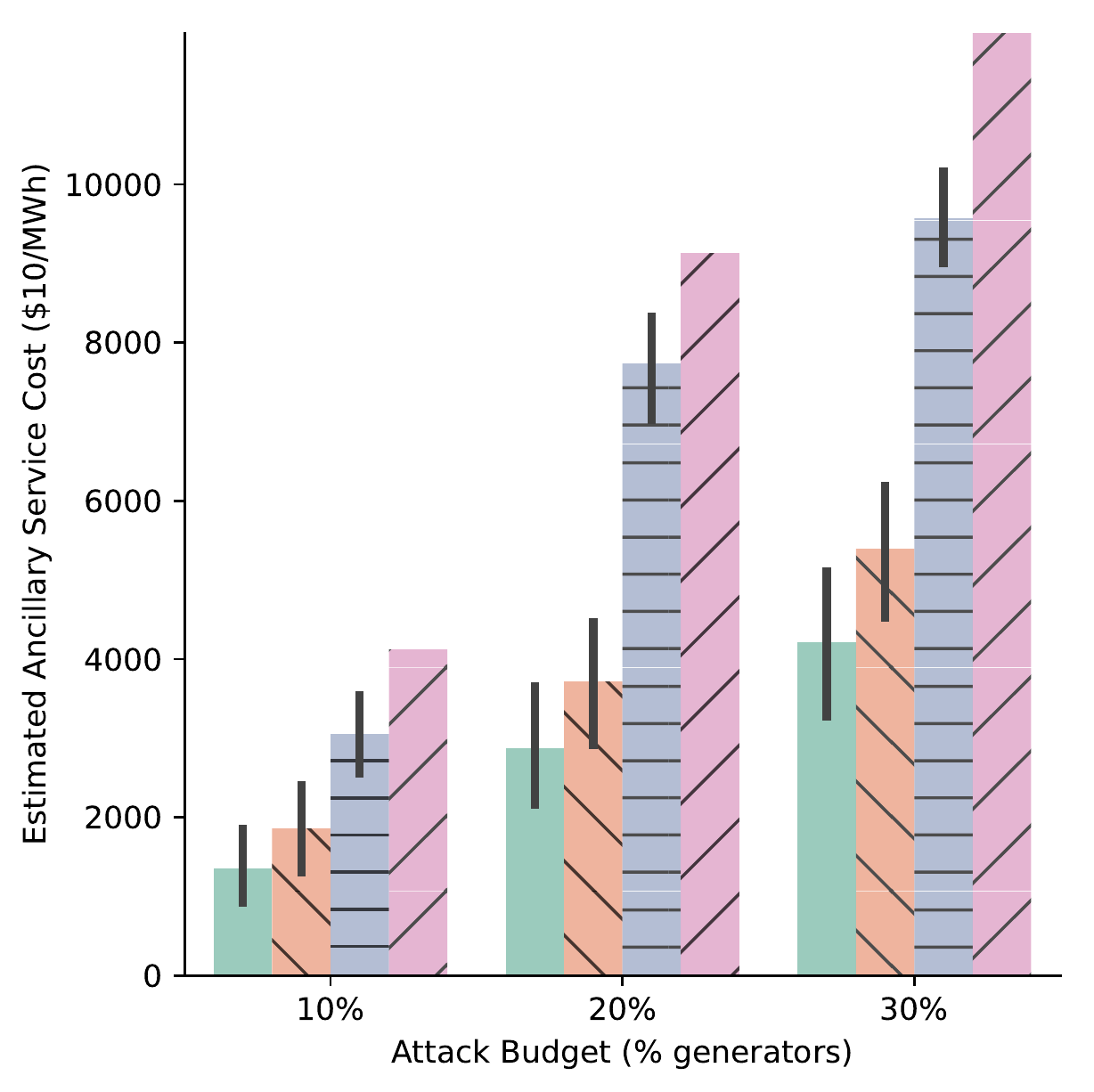}
   \includegraphics[width=.20\textwidth,height=85pt]{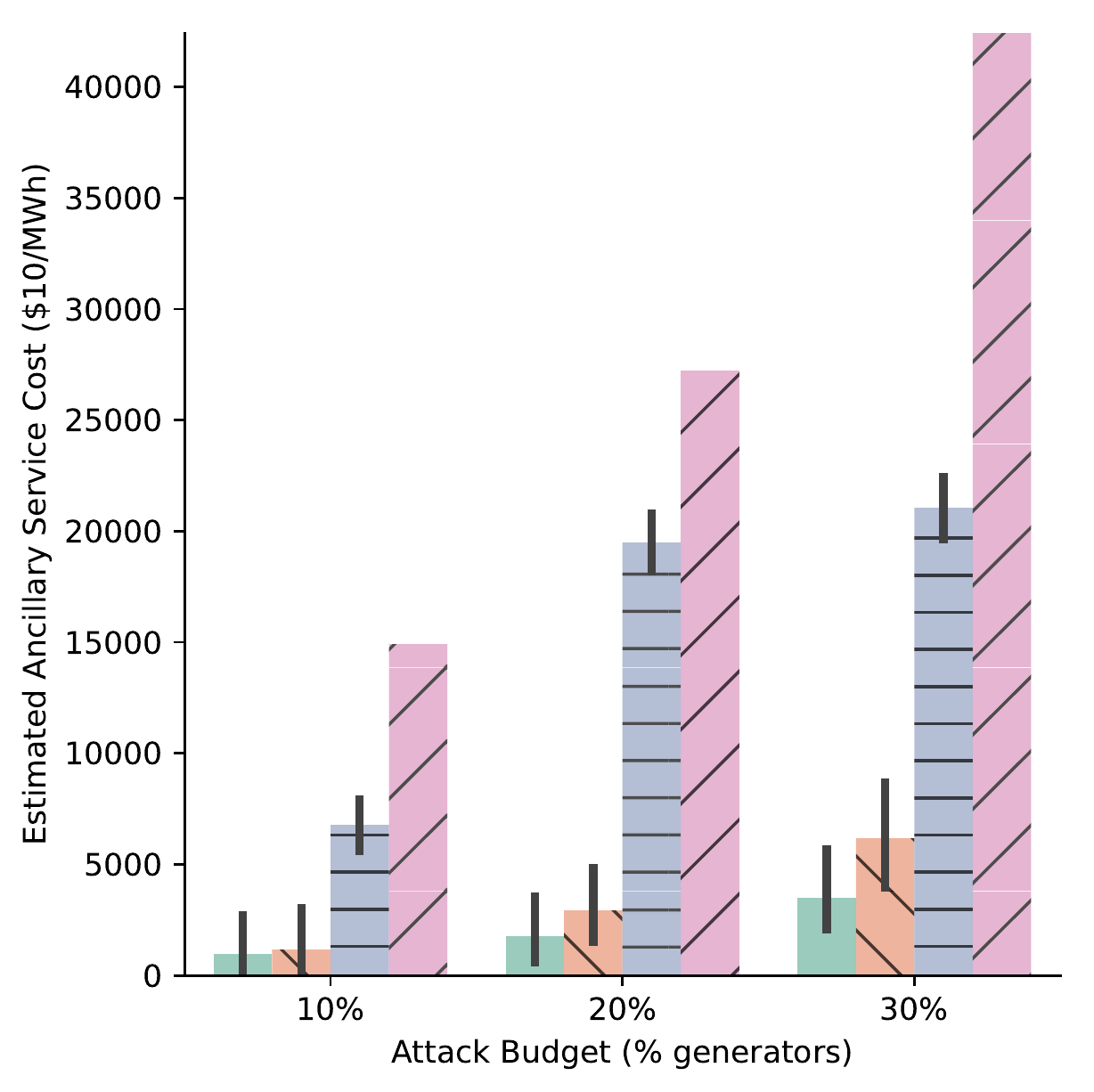}
  \caption{Ancillary service costs for the attack type on the IEEE-57 (left) and IEEE-118 (right) Networks $(\beta=0.1)$.} 
  \label{fig:attack_118}
\end{figure}

\paragraph{Attack Simulation}

It remains to study how an attacker may leverage the obfuscated
network to damage the original CIN. The attacker is given an
\emph{attack budget} $b$ denoting the percentage of generators that
can be damaged. To assess the benefits of POCIN, three type of attacks
are compared (where $k$ denote the number of generators that makes up $b\%$).

\begin{enumerate}[leftmargin=*, parsep=0pt, itemsep=0pt]

\item \emph{Random Attack}: $k$ generators are randomly selected. 

\item \emph{Obfuscated Attack}: The attacker chooses the $k$ generators with the largest dispatches in $\dot{\cin}$.

\item \emph{Fully-Informed Attack}: The attacker chooses the $k$ generators with the largest dispatches in $\cin$.
\end{enumerate}

\noindent The experiments assume estimated ancillary service costs of
\$10/MWh to serve the load that cannot be provided by the remaining
generators.  \Cref{fig:attack_118} shows the costs for each
attack type at varying of the attack budget $b \in \{10, 20, 30\}$ and the
indistinguishability value $\alpha \in \{1.0, 10.0\}, \alpha_l=\alpha
\% \times d(\cin)$ on the IEEE-57 and IEEE-118 benchmarks with
faithfulness value $\beta$ and $\alpha_v$ both set to $0.1$.  The
results average 50 simulations for each combination of parameters.
The \emph{random attacks} are used as a baseline to assess the damage inflicted by an uninformed attacker. Not surprisingly, they result in the lowest costs in each setting. In contrast, \emph{fully-informed attacks} produce the largest damages on the networks with increasing ancillary costs as the attack budget increases. 
The results for the \emph{obfuscated attacks} are particularly interesting: The obfuscation significantly reduces the power of an attacker. 
\emph{Remarkably, as the location indistinguishability values
 increase, the inflicted damages decrease and converge to that of random attacks}.  Larger indistinguishability implies more noise, and thus a higher chance for an attacker to damage less important generators. Note however that \emph{the mechanism still preserves the desired network fidelity}.

\subsection{Traffic Network Obfuscation Problem}

The scenario involves a city that would like to release its traffic
data, but is concerned about a cyber-attack on its traffic control
system. The city aims at releasing an obfuscated version of the data
that would preserve the trip durations of its commuters and would
minimize the damage of a cyber-attack on its traffic controllers
regulating traffic flows on the road segments. The case study involves
the traffic network of the city of Ann Arbor, MI and 8,000 \emph{real trips} from an origin O to a destination D (O-D pairs). For simplicity, the travel times on an
edge $e$ are given by a linear combination $d_e + \gamma t_e$ of the
distance $d_e$ and the traffic $t_e$, but the results generalize
naturally to more complex models. POCIN
obfuscates the location of the edges and their transit data, the
distance and the network being public information.
Problem P is characterized by solving a shortest path for
each commuter from her origin to her destination, using historical
traffic data. Since the shortest path problem is totally unimodular,
the bi-level program $P_{\text{BL}}$ of POCIN can be reformulated as a MIP, and thus the
problem can be solved exactly.  The experiments use a privacy loss
of $\epsilon = 1.0$, vary the \emph{indistinguishability level}
$\alpha_l$ from 1\% to 25\% of the number of nodes and select the
\emph{faithfulness level} $\beta$ of 0.1.
The implementation takes less than 20 minutes to generate the obfuscated
networks.

\paragraph{Traffic Weights \& Travel Costs}

\begin{figure}[t] \centering
\includegraphics[width=.40\textwidth,
valign=t]{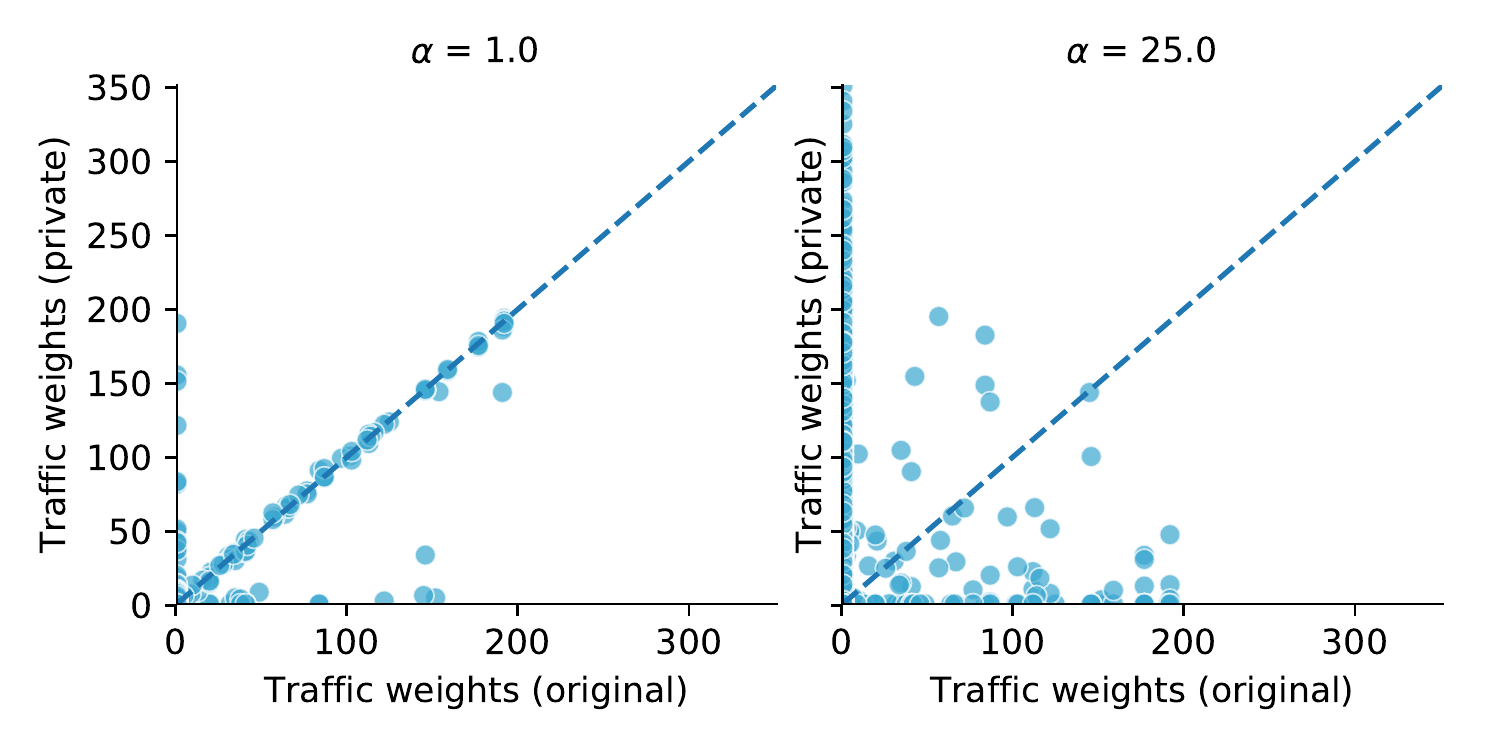}
\caption{Traffic weights on all the roads between original and
obfuscated networks with faithfulness parameters $\beta=0.1$.}
\label{fig:mobility-weights} \end{figure}

\begin{figure}[t]
\centering
\includegraphics[width=.40\textwidth, valign=t]{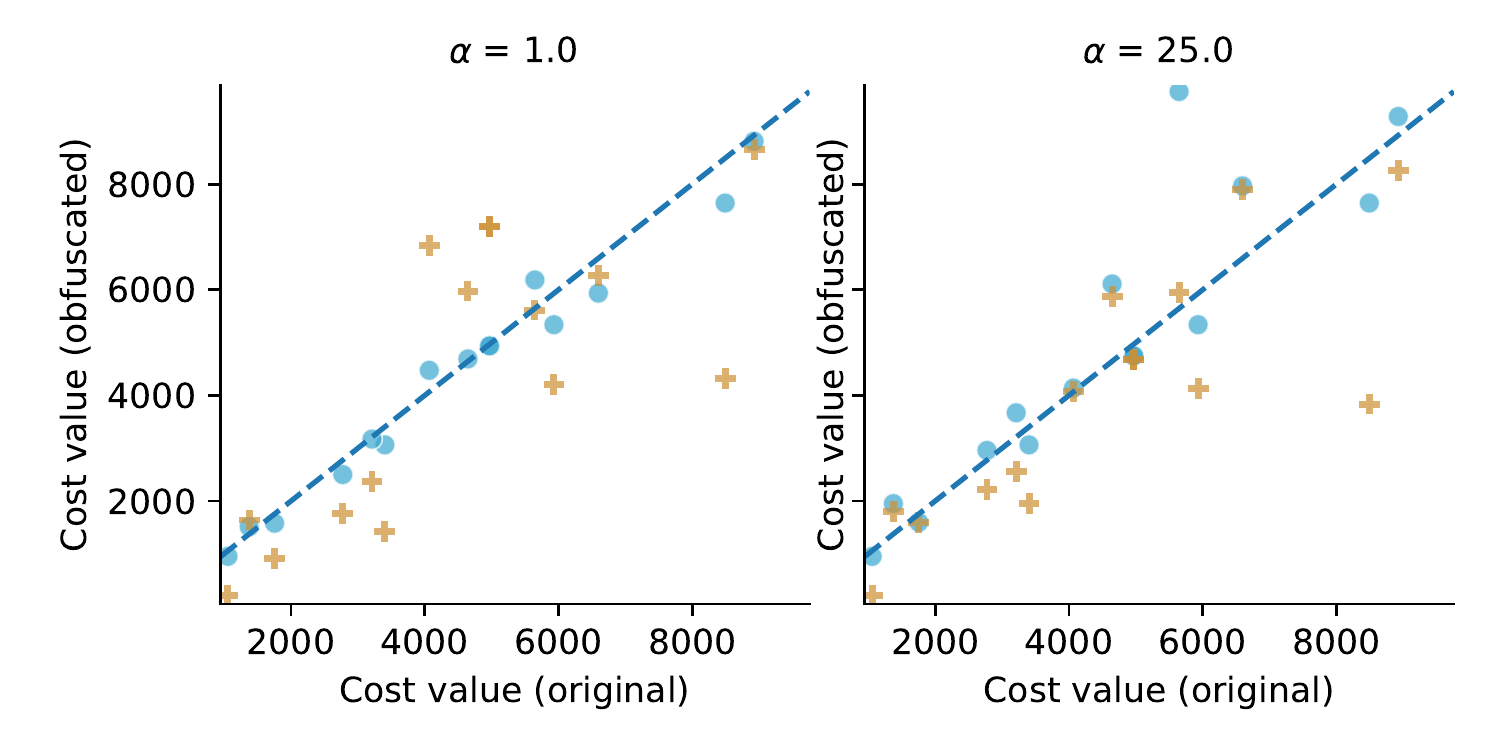}
\caption{Travel costs of 15 trips between original and obfuscated networks with faithfulness parameters $\beta=0.1$.} 
\label{fig:mobility-cost}
\end{figure}

The first experiments concern the privacy and realism of the
obfuscated network produced by POCIN. \Cref{fig:mobility-weights}
shows the travel times on all the roads in the obfuscated and original
networks.  {\em As the indistinguishability levels increase, the
  networks become significantly different}. \Cref{fig:mobility-cost}
shows the shortest paths cost of 15 O-D pairs in the obfuscated and
original networks.  Yellow/'+' dots are travel costs before Phase 3
and blue dots are the travel costs of POCIN. {\em The results show
  that POCIN preserves the shortest paths costs with high fidelity and
  that the fidelity restoration step brings significant benefits.}

\paragraph{Attack Simulation}

The second set of results concerns the obfuscation capability to
mitigate the damage of an attack. A malicious agent targets the
traffic controllers of various road segments to increase the total
travel costs of all O-D pairs. \Cref{fig:mobility-attack-cost} shows
the average increases (in \%) of the travel costs for each attack
strategy, at various attack budgets consisting of $k \in \{10, 20,
50\}$ roads, indistinguishability value $\alpha$, and with
faithfulness value $\beta = 0.1$. Larger values of $\beta$ attain
similar results.  Results are averages over 50 simulations for each
combination of parameters. The results show that random attacks are
completely ineffective, while fully-informed attacks increase travel
times substantially. In contrast, the obfuscation dramatically
decreases the damage. When the indistinguishability level increases,
the obfuscation eliminates the damage for attack budgets targeting 10
roads and only increases the travel times by less than 20\% on larger
attacks.

\begin{figure}[t]
  \centering
  \includegraphics[width=.32\textwidth]{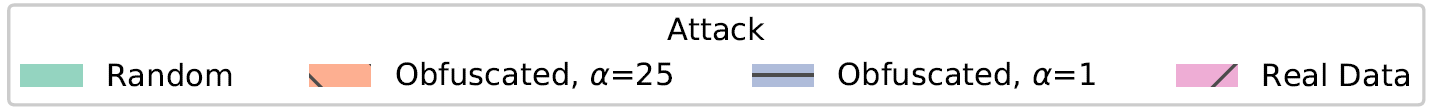}\\
  \includegraphics[width=.20\textwidth, height=80pt, valign=t]{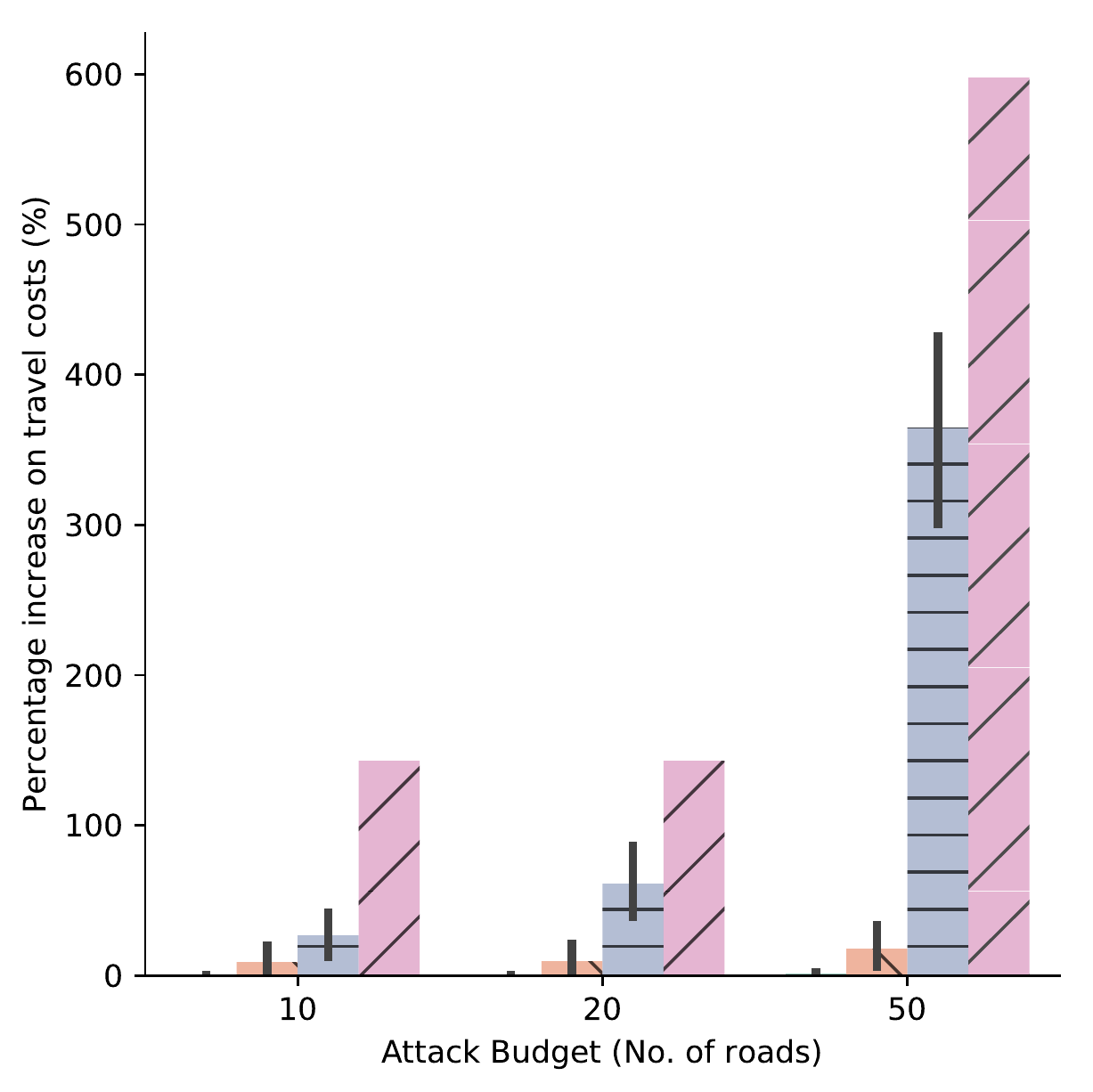}
  \caption{Average increase on travel times (in \%) of all O-D pairs with faithfulness parameters $\beta=0.1$.}
  \label{fig:mobility-attack-cost}
\end{figure}

\section{Related Work}
\label{sec:related_work}

The release of differentially private datasets that protects the
\emph{value} of the data is a challenging task that has been studied
by several authors. Often the released data is generated from a data
synopsis in the form of a noisy
histogram~\cite{li:14,qardaji:14,fioretto:AAMAS-18}.  Protecting locations privacy typically uses the framework
of \emph{geo-indistinguishability} \citep{andres2013geo}, in which
locations are made indistinguishable in the Euclidean plane.

The release of differentially private networks has also been studied in the literature: It primarily focuses on social networks
\cite{Proserpio:14,Blocki:13,kasiviswanathan2013analyzing}. These
methods focus on either protecting the privacy of the \emph{node} or
the \emph{edge} and ensure that the output distribution does not
change significantly when a node (resp.~edge) and all its adjacent
edges (resp.~nodes) are added to, or removed from, the graph. The goal
of these methods is not to preserve the original network topology.  In
contrast to these approaches, POCIN focuses on the release of
privacy-preserving network data that protects locations and values of
the network components while preserving the network topology and
fidelity. As shown by the experimental results, this places additional
requirements on the mechanism, including the need for POCIN to use
optimization to redistribute the noise and ensure that the data
release admits a realistic solution to the optimization problem.

\emph{The evaluation of privacy-preserving techniques with respect to
malicious attacks is also a novel contribution of this
paper}. Traditional research on attack detection and prevention on
cyber-physical systems is not concerned with privacy and data release. For instance \citep{ghafouri2018adversarial} uses
regression to detect anomalous sensor readings
and \citep{junejo2016behaviour} focuses on
predicting attacks observing real and simulated data.  Finally,
bi-level problems have been used in security applications of AI
\cite{pita2009using,jain2011double,nguyen2016conquering,zhao2017efficient}.

\section{Conclusions}
\label{sec:conclusions}

This paper presented a privacy-preserving scheme for the release of
\emph{Critical Infrastructure Networks} (CINs). The proposed
\emph{Privacy-preserving Obfuscation mechanism for CIN} (POCIN)
obfuscates values and locations of sensitive network elements
combining several differential privacy mechanisms with a bi-level
optimization problem.  It does so without altering the CIN topology
and ensuring that the released obfuscated network preserves the
fundamental properties of an optimization problem of interest.

The paper proposes exact and relaxation solutions for solving the
bi-level optimization problems. The POCIN mechanism was tested on
realistic test cases for the power grid and transportation system.
The results show that POCIN is effective in obfuscating values and
locations of the network parameters. Importantly, the result also
illustrates the effectiveness of POCIN to deceive malicious agents
that exploit the data release to produce as much damage as possible to
the CIN.

\paragraph{Acknowledgments} The authors would like to thank Kory Hedman for extensive discussions. The authors are also grateful to the anonymous reviewers for their valuable comments. This research is partly funded by the ARPA-E Grid Data Program under Grant 1357- 1530.

\newpage

\bibliographystyle{named} \bibliography{opf_bib}

\newpage
\setcounter{theorem}{5}
\setcounter{corollary}{0}
\setcounter{figure}{10}

\appendix
\section{Missing Proofs}

First, we review some important properties of DP.

\emph{Composability} ensures that a combination of differentially  private algorithms preserve differential privacy.


\begin{theorem}[Parallel Composition] 
\label{th:par_composition} 
Let $D_1$ and $D_2$ be disjoint subsets of $D$ and $\cA$ be an
$\epsilon$-DP algorithm.  Computing $\cA(D
\cap D_1)$ and $\cA(D \cap D_2)$ satisfies $\epsilon$-DP.
\end{theorem}

\emph{Post-processing immunity} ensures that privacy guarantees are
preserved by arbitrary post-processing steps. 

\begin{theorem}[Post-Processing Immunity] 
\label{th:postprocessing} 
Let $\cA$ be an $\epsilon$-DP algorithm and $g$ an arbitrary mapping from the set of possible outputs $\cO$ to an arbitrary set. Then, $g \circ \cA$ is $\epsilon$-DP.
\end{theorem}

  \begin{corollary}
    \label{thm:pahse2}
    Given $\alpha_\ell > 0$, Phase 1 of POCIN achieves $\alpha_\ell$-location-indistinguishability. 
  \end{corollary}

\begin{proof}
    Observe that the Location Obfuscation phase (Phase 1) combines the
    exponential mechanism with a post-processing step.  Therefore, the
    result follows by Theorem \ref{th:m_exp-1}, parallel composition
    (Theorem \ref{th:par_composition}), and post-processing immunity
    of differential privacy (Theorem \ref{th:postprocessing}).
  \end{proof}

\begin{corollary}
  \label{thm:pahse1}
  Given $\alpha_v > 0$, Phase 2 of POCIN achieves $\alpha_v$-value-indistinguishability. 
\end{corollary}

\begin{proof}
  Observe that the Value Obfuscation phase (Phase 2) is an application of the Laplace Mechanism with scaling factor $\alpha_v / \epsilon$ (Equation~(4): main article). The result follows from Theorem 2 and by observing that the Laplace Mechanism with parameter $\alpha / \epsilon$ achieves $(\epsilon, \alpha)$-indistinguishability \cite{chatzikokolakis2013broadening}.
\end{proof}

\setcounter{theorem}{2}
\begin{theorem}
  POCIN is $(\epsilon, \alpha_p, \alpha_v)$-indistinguishable. 
\end{theorem}
 
\begin{proof}
  From Corollaries \ref{thm:pahse2} and \ref{thm:pahse1} the application of Phase 1 and Phase 2 of POCIN achieves $(\epsilon, \alpha_\ell)$-indistinguishability for locations and $(\epsilon, \alpha_v)$-indistinguishability for values, respectively. 
  Therefore, from Equation~(3) (main article)
  the Phase 1 and 2 of POCIN are location and value indistinguishable and achieve $(\epsilon, \alpha_p, \alpha_v)$-indistinguishability.   Finally, Phase 3 uses exclusively differentially private outputs and public information. Therefore the result follows from post-processing immunity of differential privacy (Theorem \ref{th:postprocessing}).
\end{proof}

The proof of the next theorem is similar to the proof of (Theorem 5) in \cite{fioretto:CPAIOR-18}.

\begin{theorem}
\label{theorem:factorA}
The error induced by POCIN on the CIN node values is bounded by the inequality:
  $
  \| \dot{\bv}^* - \bv \|_2 \leq 2
  \| \tilde{\bv} - \bv \|_2.
  $
where $\dot{\bv}^*$ is the optimal solution to problem $P_{\text{BL}}$.  
\end{theorem}
\begin{proof}
The results follows from the fact that $\bv$ is feasible for the lower-level problem, the optimality of $\dot{\bv}^*$ for $P_{\text{BL}}$, and the triangle inequality on norms.
\end{proof}

\begin{theorem}
  When $P_{\text{CL}}$ is used for fidelity restoration and is convex, the error induced by POCIN on the CIN node values is bounded by the inequality: 
  $$
  \| \dot{\bv}^* - \bv \|_2 \leq \| \tilde{\bv} - \bv \|_2.
  $$
  where $\dot{\bv}^*$ is the solution to the problem $P_{\text{CL}}$.
\end{theorem}
  \begin{proof}
  Let $\CC_P \in \RR^n$ be the feasible set of $P_{\text{CL}}$, i.e.,
  $\CC_P = \{ \bv | \exists x \mbox{ such that } (n2) - (n3) \text{
    hold} \}$.  By definition, $\CC_P$ is closed, convex, and
  non-empty, since $\bv$ is a feasible solution.  Since $\bv$ is the original vector, then $\bv \in
  \CC_P$.  The addition of noise to $\bv$ produces a new vector
  $\tilde{\bv} \in \RR^n$ that may or may not satisfy the problem
  constraints.
\begin{itemize}
  \item Case 1: $\tilde{\bv} \in \CC_P$.  That is, $\tilde{\bv}$
    satisfies the constraints of $P_{\text{CL}}$.  The minimizer for
    $P_{\text{CL}}$ is thus the vector $\dot{\bv}^* = \tilde{\bv}$
    with objective value 0. Therefore:
  $$
  \| \dot{\bv}^* - \bv \|_2 = \| \tilde{\bv} - \bv \|_2.
  $$

\item Case 2: $\tilde{\bv} \not\in \CC_P$, i.e., $\tilde{\bv}$ lies outside the feasible region $\CC_P$.
    Problem $P_{\text{CL}}$ can be seen as a \emph{projection problem} onto a convex set, i.e.,
  $$
    \text{Proj}_{\CC_p}(x) = \argmin_{v \in \CC_p} \| x - v\|_2.
    $$ The variational characterization of projection implies that the
    solution $\dot{\bv}^*$ satisfies $\langle \bv -
    \dot{\bv}^*,\tilde{\bv} - \dot{\bv}^*\rangle \leq 0$ and the
    result follows from the triangle inequality on norms.
\end{itemize}
\end{proof}

\end{document}